\documentclass[a4paper,UKenglish,runningheads,11pt]{llncs}

\usepackage{tabularx,booktabs,multirow,delarray,array}
\usepackage{graphicx,amssymb,amsmath,mathtools}
\usepackage{enumerate}
\usepackage[ruled,vlined,linesnumbered]{algorithm2e}
\usepackage{subfigure}
\usepackage{latexsym}
\usepackage{lineno}
\usepackage{hyperref}
\usepackage{url}
\usepackage{upgreek}
\usepackage{textgreek}
\usepackage{lipsum}
\usepackage{ntheorem}
\usepackage{fullpage}
\usepackage[colorinlistoftodos]{todonotes}

\newenvironment{proof}{\par\noindent{\bf Proof:}}{\mbox{}\hfill$\qed$\\}
\newtheorem*{theorem-non}{Theorem}

\newlength{\bibitemsep}
\newlength{\bibparskip}
\let\oldthebibliography\thebibliography
\renewcommand\thebibliography[1]{
  \oldthebibliography{#1}
  \setlength{\parskip}{\bibitemsep}
  \setlength{\itemsep}{\bibparskip}
}

\newcommand{\ignore}[1]{ }

\newcounter{rem}
\def\etal{\textsl{et~al.}}

\def\qed{\hbox{\rlap{$\sqcap$}$\sqcup$}}

\begin{document}

\title{Local Routing on a Convex Polytope in $\mathbb{R}^3$}
\titlerunning{Local routing on a convex polytope in $\mathbb{R}^3$} 

\author{
Sreehari Chandran\inst{1}
\and
R. Inkulu\inst{1}
}

\institute{
Department of Computer Science and Engineering\\
Indian Institute of Technology Guwahati\\
\email{\{c.sreehari,rinkulu\}@iitg.ac.in}
}


\maketitle

\pagenumbering{arabic}
\setcounter{page}{1}

\begin{abstract}
Given a convex polytope $P$ defined with $n$ vertices in $\mathbb{R}^3$ and a parameter $\epsilon \in (0, 1)$, this paper presents an algorithm to preprocess $P$ to compute routing tables at every vertex of $P$ so that a data packet can be routed on the boundary $\partial P$ of $P$ from any vertex $s$ of $P$ to any other vertex $t$ of $P$.
At every vertex $v$ of $P$ along the routing path on $\partial P$ from $s$, until the packet reaches $t$, the next hop is determined using the routing tables at $v$ and the information stored in the packet header.
In $O(n \min(n^2, \frac{1}{\epsilon^7} \lg{n}))$ time, the preprocessing algorithm computes a routing table at every vertex of $P$ of amortized size $O((\min(n, \frac{1}{\epsilon^{3/2}}))\lg{n})$ bits.
If the shortest distance between $s$ and $t$ on $\partial P$ is $d(s, t)$, then the routing path produced by this algorithm has length at most $\frac{8+\epsilon}{\sin{\theta_m}}(D+d(s,t))$.
Here, $D$ is the maximum length of the diagonal of any cell when $\partial P$ is partitioned into $\frac{1}{\epsilon^3}$ geodesic cells of equal size, and $\theta_m$ is half the minimum angle between two edges bounding any face of $\partial P$.
\end{abstract}

\section{Introduction}
\label{sect:intro}

For any plane $h$ in $\mathbb{R}^3$, excluding $h$ from $\mathbb{R}^3$ defines two regions, each of which is called an {\it open half-space} defined by $h$.
A half-space defined by $h$ together with $h$ is a {\it closed half-space} of $h$.
The intersection of a finite set of closed half-spaces in $\mathbb{R}^3$, each defined by a plane in $\mathbb{R}^3$, is called a {\it convex polytope}.
The convex hull of a finite set of points in $\mathbb{R}^3$ is also a convex polytope.
Minkowski~\cite{journals/teubner/Minkowski1896} proved that these two definitions of convex polytopes are equivalent. 
Indeed, these definitions extend to hyperplanes located in high-dimensional spaces as well.
The boundary $\partial P$ of $P$ consists of all points $p \in P$ such that for any open ball $B$ of infinitesimally small radius centered at $p$, $B \cap P \ne B$.
In this paper, we assume the convex polytope is full-dimensional and is located in $\mathbb{R}^3$.

For any data packet originating at any (source) vertex of $P$, the {\it routing problem} seeks to find a path on $\partial P$ for that packet to proceed to its destination vertex of $P$. 
At every vertex $v$ along that path, the next hop for the packet is determined from the routing table stored at $v$ and the routing information in the packet header.
This requires preprocessing $P$ to construct routing tables at its vertices. 
Naturally, the routing information in the packet header includes the destination vertex identifier.
The typical objectives of a routing problem include minimizing the following: 
(i) for any two vertices $s$ and $t$ of $P$, the maximum ratio of the length of the routing path from $s$ to $t$ computed by the routing algorithm to the (geodesic) distance from $s$ to $t$, which is called the {\it stretch} of the routing path,
(ii) the maximum space required for routing tables at any vertex of $P$,
(iii) the time to preprocess $P$ to compute routing tables at the vertices of $P$,
(iv) the maximum size of the label associated with any vertex of $P$ for identifying vertices, and
(v) the maximum number of bits required for the routing information in the header of the packet.

A local routing algorithm determines the next hop based on the location of $t$, the location of the vertex $v$ at which the packet is residing, and the ($k$-)neighbourhood of $v$, for a fixed $k$.
In doing so, the algorithm may use the information stored at $v$.
Specifically, in the case of local routing on $\partial P$, the next hop for any packet at any intermediate vertex $v$ of $\partial P$ along the routing path must be a vertex $v'$ that is a neighbour of $v$.
That is, there must exist an edge of $P$ between $v$ and $v'$ so that $v$ can forward the packet to $v'$.

In the plane, a simple polygon with pairwise-disjoint simple polygonal holes is called a polygonal domain.
Since no path can intersect with the interior of any of these holes, these holes are known as obstacles of the polygonal domain.
The routing problem for polygonal domains seeks to route data packets from one vertex to another.
This involves preprocessing the input polygonal domain to build the data structures required at its vertices. 
A pairwise-disjoint set of convex polygonal obstacles in the plane is called a convex polygonal domain.
Banyassady~\etal~\cite{journals/cgta/BanyaCKMR20} devised a routing scheme for simple polygonal domains. 
For convex polygonal domains, Inkulu and Kumar~\cite{journals/ijfcs/InkuluKum24} devised a routing scheme that improves the preprocessing time and the size of routing tables of the algorithm given in \cite{journals/cgta/BanyaCKMR20}.
Gaur and Inkulu~\cite{journals/joco/GaurInkulu25} devised a routing scheme with improved preprocessing time for simple polygons.
The routing scheme presented in Gaur and Inkulu~\cite{journals/joco/GaurInkulu26} improves both the total size of routing tables and the preprocessing time of the algorithm in \cite{journals/cgta/BanyaCKMR20} for convex polygonal domains.

For the problem of computing an optimal shortest path on $\partial P$ between two vertices of a polytope, Sharir and Schorr~\cite{journals/siamcomp/SharirS86}, Chen and Han~\cite{journals/ijcga/ChenH96}, and Schriber and Sharir~\cite{journals/dcg/SchreiberS08} devised algorithms with time complexities $O(n^3\lg{n})$, $O(n^2\lg{n})$, and $O(n\lg{n})$, respectively.
Hershberger and Suri~\cite{journals/comgeo/HershbergerS98b} devised a simple $O(n)$ time algorithm to output a $2$-approximate shortest path.
Agarwal~\etal~\cite{journals/jacm/AgarwalPSV97} developed an $O(n+(\lg{\frac{1}{\epsilon}}) + \frac{1}{\epsilon^3})$ time algorithm to output a $(1+\epsilon)$-approximate shortest path.
This result is efficient if only the $(1+\epsilon)$-approximate distance is needed, not the path itself.
Computing shortest paths to all the vertices of a convex polytope from a fixed source vertex was also studied in \cite{journals/ijcga/ChenH96,journals/siamcomp/MitchellMP87,journals/siamcomp/SharirS86}.
Especially, after preprocessing the convex polytope, the algorithm by Har-Peled~\cite{journals/siamcomp/Har-Peled99} computes a shortest path from a fixed source to any vertex in $O(\lg{\frac{n}{\epsilon}})$ time.
In addition, there is work on computing shortest paths on polytopes which are not necessarily convex.
For general polytopes in $\mathbb{R}^3$, the most notable result is by Mitchell~\etal~\cite{journals/siamcomp/MitchellMP87}, which computes an optimal single-source shortest path from any given vertex to all other vertices of the input polytope by progressing a Dijkstra wavefront over $\partial P$.
Mitchell~\cite{coll/hb/Mitch17} surveys shortest path algorithms in geometric domains.

A graph is a geometric graph if it is embedded in the plane.
Yao-graphs~\cite{journals/siamcomp/Yao82} and \textTheta-graphs~\cite{conf/stoc/Clarkson87} are commonly used in computing geometric spanner networks. 
The significant works for routing in geometric graphs include \cite{journals/jocg/BoseKRV18,journals/tcs/BoseMorin04,conf/distrcomp/HassPeleg00,conf/cccg/Kranakis99}.
The routing algorithms for \textTheta-graphs and half-\textTheta$_6$-graphs are considered in Bose~\etal~\cite{journals/jocg/BoseCD20,journals/siamjc/BoseFRV15}, \textTheta$_4$-graphs are considered in \cite{conf/soda/BoseCHS19}, and Yao-graphs in Bose~\etal~\cite{conf/isaac/BoseHST24}.
And the algorithms for routing on Delaunay triangulations are explored in \cite{conf/esa/BonichBDDHS18,journals/dcg/BonichPBDPR17,journals/siamjc/BoseFRV15}.
Competitive local routing schemes for geometric graphs are presented in \cite{conf/esa/BonichBDDHS18,journals/dcg/BonichPBDPR17,journals/siamjc/BoseFRV15,journals/jocg/BoseFRV17,conf/isaac/BoseHST24,journals/tcs/BoseMorin04}.
These local routing algorithms do not use any routing tables.
For both convex and general polytopes, Kapoor and Li~\cite{conf/isaac/KapoorL09} presented algorithms to compute a geodesic Steiner spanner network on the surface of the polytope, with the set comprising vertices of the polytope being a subset of the node set of the spanner network.

The routing problem is also popular in abstract graphs.
In this case, a packet needs to be routed from any source vertex of the input graph $G$ to the packet's destination vertex along a path in $G$. 
The preprocessing algorithm computes routing tables at every vertex of $G$.
A naive preprocessing algorithm could compute all-pairs shortest paths in $G$, and for a vertex $u'$ adjacent to $u$ in a shortest path from $u$ to $v$, the algorithm may store the tuple $(u', v)$ as a routing table entry at $u$.
Though this routing scheme yields an optimal stretch, the routing table at any node will have $O(n)$ entries.
Early work on the routing problem in graphs focused on routing schemes for the special case of general graphs such as trees \cite{conf/icalp/FragGavo01,journals/comp/SantoroK85}, planar graphs \cite{journals/jacm/Thorup04}, unit disk graphs \cite{journals/algorithmica/KaplanMRS18,journals/cgta/YanXD12}, and for networks of low doubling dimension \cite{journals/talg/KonjevodRX16}.
For every $k \ge 1$, the routing scheme presented by Awerbuch~\etal~\cite{journals/jalgo/AwerbuchBLP90} guarantees a stretch factor of $O(k^2)$ and requires storing $\tilde{O}(k n^{1/k} (\lg{D}))$ bits of routing information per vertex of the input graph with diameter $D$.
Cowen~\cite{journals/jalgo/Cowen01} designed a routing scheme that has a multiplicative stretch $3$, packet headers of size $O(\lg{n})$, and routing tables of size $\tilde{O}(n^{2/3})$.
Thorup and Zwick~\cite{conf/spaa/ThorupZwick01,journals/jacm/ThorupZwick05} presented a routing scheme that uses $\tilde{O}(kn^{1/k})$ bits of space at every node and has a multiplicative stretch of $4k-5$, for every $k \geq 1$.
Chechik~\cite{conf/podc/Chechik13} devised a scheme using $\tilde{O}(n^{1/k} \lg{D})$ bits of space at every node and having a stretch $c k$, for some $c < 4$ and for sufficiently large integer $k$.
Roditty and Tov~\cite{conf/podc/RodittyTov15} designed a routing scheme using $\tilde{O}(\frac{1}{\epsilon}n^{1/k} \lg{D})$ bits of space at each node and having a stretch $4k-7+\epsilon$, for every integer $k \ge 1$.
In all these results, $D$ is the diameter of the input graph and $\tilde{O}()$ hides powers of $\lg{n}$.
Peleg and Upfal~\cite{journals/jacm/PelegUpfal89} have shown that any routing scheme with constant stretch factor needs to store $\Omega(n^c)$ bits per node for some constant $c > 0$.

\subsection*{\bf Preliminaries}
\label{subsect:prelim}

We denote the input convex polytope located in $\mathbb{R}^3$ by $P$.
The boundary of $P$ (defined earlier) is denoted by $\partial P$, and the relative interior $P \backslash \partial P$ of $P$ is denoted by $rint(P)$.
Analogously, for any line segment $\ell$, $\ell$ excluding its endpoints is the relative interior of $\ell$, denoted by $rint(\ell)$.
We assume each face of $P$ is triangulated, and these triangles together triangulate $\partial P$.
The set comprising the vertices of $P$ is denoted by $V$, and $|V|$ is denoted by $n$.
We say that the two vertices $u$ and $v$ of $P$ are neighbours whenever there is a triangulation edge joining them. 
Each vertex is associated with a routing label that identifies it.
To distinguish from the vertices of $P$, every point in $\mathbb{R}^3$ that is not necessarily a vertex of $P$ is called a point, and the vertices of graphs are called nodes.

For any point $p$ on $\partial P$, the projection of $p$ on plane $\eta$ is a point $p'$ located on $\eta$ such that among all the line segments with one endpoint on $\eta$ and the other point being $p$, the line segment $pp'$ is the shortest.
For any face $f$ of $P$, the normal $n_f$ to $f$ is a vector orthogonal to $f$, originating from any point on $f$ with $n_f \cap rint(P) = \emptyset$.
By saying a plane $h$ in $\mathbb{R}^3$ is stored in any entry of a routing table $\rho$, we mean a point $p$ located on $\partial P$ and two non-parallel vectors located on $h$ both originated from $p$ together defining $h$, and the tuple consisting of $p$ and two points respectively located on each of these vectors is stored in $\rho$.
When a packet is at a vertex $v$, and its next hop along the routing path is $w$, we say the packet is {\it forwarded} from $v$ to $w$.
Otherwise, the packet is said to have been {\it routed} from $v$ to $w$.

The Euclidean distance between any two points $p$ and $q$ in $\mathbb{R}^3$ is the distance along the line segment joining $p$ and $q$, and is denoted by $\Vert pq \Vert$.
For an edge-weighted graph $G$, the shortest path distance between any two nodes $u$ and $v$ of $G$  is denoted by $d_G(u, v)$.
For any polygonal path $\lambda$ on $\partial P$, the length of $\lambda$, denoted by $|\lambda|$, is the sum of the lengths of all line segments belonging to that path.
A geodesic path between any two points $p, q \in \partial P$ is a polygonal path located on $\partial P$ consisting of line segments such that every endpoint of every such segment (except for $p$ and $q$) is either a point on an edge of $P$ or a vertex of $P$ such that the path length cannot be improved by local adjustments.
For any two points $p, q$ located on $\partial P$, among all the geodesic paths between $p$ and $q$, a path that has the minimum length is called a shortest (geodesic) path between $p$ and $q$.
The distance between $p$ and $q$ located on $\partial P$ is the length of a shortest path between $p$ and $q$.
The {\it stretch} of a routing scheme on $\partial P$ is the maximum, over all pairs of vertices $u$ and $v$ of $P$, ratio between the length of the routing path from $u$ to $v$ output by that routing scheme to the distance between $u$ and $v$ on $\partial P$. 

Let $r'$ and $r''$ be two non-parallel rays with origin at point $p$.
Let $\overrightarrow{v_1}$ and $\overrightarrow{v_2}$ be the unit vectors along rays $r'$ and $r''$, respectively.
A {\it cone} $C_p(r', r'')$ is the set of points defined by rays $r'$ and $r''$ such that a point $q \in C_p(r', r'')$ if and only if $q$ can be expressed as a convex combination of vectors $\overrightarrow{v_1}$ and $\overrightarrow{v_2}$ with positive coefficients.
When the rays are evident from the context, we denote $C_p(r', r'')$ by $C_p$ or $C$.

\vspace{-0.05in}
\subsection*{\bf Our Contributions}
\label{subsect:contrib}

Given a convex polytope $P$ defined with $n$ vertices and a real number $\epsilon \in (0, 1)$, in the preprocessing phase, our algorithm computes a routing table at every vertex of $P$.
Our preprocessing algorithm approximates $P$ by another convex polytope $P'$ whose complexity is independent of $n$ but is dependent only on the input parameter $\epsilon$, though our algorithm does not compute $P'$ itself.
As part of this, following Dudley~\cite{journals/apprxth/Dudley74}, Agarwal~\etal~\cite{journals/jacm/AgarwalPSV97}, and Kapoor and Li~\cite{conf/isaac/KapoorL09}, we partition $\partial P$ into $\delta$-patches.
Each $\delta$-patch is a contiguous section of $\partial P$ and the angle between any two faces belonging to a $\delta$-patch is upper bounded by $\delta$.
Later, in the analysis, we express $\delta$ in terms of $\epsilon$.
For each $\delta$-patch $\delta_i$ of $P$, we select one face $f_i$ belonging to $\delta_i$.
For every such $f_i$, let $h_i$ be the half-space defined by the plane containing $f_i$ such that $h_i$ contains $P$. 
Then, $P' = \bigcap h_i$ is a convex polytope, and $P'$ contains $P$.

Next, as described below, for every $\delta$-patch $\delta_i$, our algorithm deterministically selects a subset $R_i$ of the vertices of $\delta_i$.
For computing $R_i$, our algorithm projects all the vertices of $\delta_i$ to the plane $\delta_i'$ containing the representative face of $\delta_i$.
For every $i$, $\delta_i'$ is partitioned into $\frac{1}{\epsilon}$ grid cells.
And, among all the vertices of $P$ which get projected to any grid cell $g_{ij}'$, one arbitrary point $r_{ij}'$ is chosen as the representative of cell $g_{ij}'$.
If no vertex of $P$ is projected to a grid cell, then that grid cell will not have any representative vertices.
A vertex $r_{ij}$ of $\delta_i$ is defined as a representative vertex of $\delta_i$ whenever the point of projection $r_{ij}'$ of $r_{ij}$ on $\delta_i'$ is a representative of some grid cell located on $\delta_i'$. 
Thus, for every patch $\delta_i$, our algorithm defines at most $\frac{1}{\epsilon}$ representative vertices.
The representative vertex $r_{ij}$ is the hub for the set $S_{ij}$ of vertices of $P$ that gets projected to $g_{ij}'$.
That is, each vertex in $S_{ij}$ other than $r_{ij}$ has routing information to forward any packet to only $r_{ij}$, and $r_{ij}$ can route the packet to any other vertex in $S_{ij}$.
In addition, among all the vertices of $S_{ij}$, only $r_{ij}$ has the information to route packets to representative vertices that are located on $\delta_i$ or to the representative vertices located elsewhere on $\partial P$.

By the end of the preprocessing phase, any representative vertex $r_{ij}$ of $P$ has two kinds of entries in its routing table: one to route to any vertex whose representative vertex is $r_{ij}$, and the other to route to any representative vertex of $P$.
Any non-representative vertex $v$ of $P$ has only one routing entry, which corresponds to routing to its representative vertex $r_{ij}$.
Specifically, the routing table at $v$ stores a plane orthogonal to $\delta_i'$ containing both $v$ and $r_{ij}$. 
And at $r_{ij}$, for every vertex $v \in S_{ij} \backslash \{r_{ij}\}$, we store a plane orthogonal to $\delta_i'$ containing both $r_{ij}$ and $v$.
Analogously, for any patch $\delta_i$ and for any two representative vertices $r_{ij}$ and $r_{ik}$ belonging to $\delta_i$, the routing table at $r_{ij}$ stores a plane orthogonal to $\delta_i'$ containing $r_{ij}$ and $r_{ik}$.
Indeed, our preprocessing algorithm ensures that every routing entry in any routing table is to a point that belongs to the same patch.
For any plane $\eta$, the intersection of $\eta$ with $\delta_i$ defines a polygonal path $\lambda$ on $\delta_i$; the routing path computed by our algorithm uses a subset of edges of $P$ that intersect $\lambda$. 

To facilitate routing from any representative vertex $u$ located on patch $\delta_i$ to any representative $v$ located on patch $\delta_j$ with $i \ne j$, we compute a geodesic Steiner spanner on each of the faces of $P'$.
For every $i$, let $R_i'$ be the set comprising all the points resulting from projecting representative vertices of $\delta_i$ on the plane containing $\delta_i'$. 
By following Kapoor and Li~\cite{conf/isaac/KapoorL09}, for every $i$, we compute a \textTheta-graph based geodesic Steiner spanner $G_i'$ whose node set is a superset of $R_i'$.
We denote $\bigcup_i G_i'$ with $G'$.
For every two representative vertices $u, v$ of $P$ respectively located on patches $\delta_i$ and $\delta_j$, with their respective projections $u', v'$ located on $P'$, while considering the routing of a packet from $u'$ to $v'$ in $G'$, we find the node $w'$ in $G'$ the packet from $u'$ gets forwarded when the routing algorithm in Thorup and Zwick~\cite{conf/spaa/ThorupZwick01} is used.
Then we store in the routing table at $u$, for the entry corresponding to $v$, the plane orthogonal to $\delta_i'$ containing $u'$ and $w'$.
If $w'$ corresponds to a Steiner node $r$ on an edge $e$ of $P$, since Steiner points on $\partial P$ cannot have a routing table, we mark both the vertices of $P$ that are incident to $e$.
One of these marked vertices receives packets on behalf of $w$, and both the marked vertices have sufficient information in their routing tables to forward such packets further.

The routing tables computed during the preprocessing phase help route data packets between the vertices of $P$.
Specifically, at any vertex $v$ of $P$, the next hop for the packet is determined using the routing table stored at $v$ together with the data in the packet header.
For the packet at any vertex $u$, the local routing algorithm proposed here guarantees the next hop of the packet is a vertex $v$ on $\partial P$ such that there is an edge of $P$ joining $u$ and $v$.
Let $\eta$ be the plane used to route the packets from vertex $u$ of $\delta_i$ to a point $p$ located on $\delta_i$. 
Then, in the routing phase, our algorithm uses a greedy strategy to route, considering the polygonal path $\eta \cap \delta_i$.
Significantly, the routing path is a subgraph of edges of $P$ that are intersected by $\eta \cap \delta_i$.
The approximation resulting from the greedy algorithm is analyzed using a stays ahead argument.

With this research, we have explored several ideas in the context of geometric routing.
These include computing a sketch $P'$ of the input convex polytope $P$, selecting vertices as representatives via deterministic sampling by setting up a grid over $P'$, setting up a two-level hierarchy over the vertices of $P$ and designing respective routing schemes, computing a geodesic Steiner spanner on each face of $P'$, using that spanner to compute routing tables at the vertices of $P'$, enhancing those tables to build the corresponding routing tables at the sampled vertices of $P$, and the greedy routing scheme to ensure local routing.
To our knowledge, this is the first result for local routing on polytopes.
The following theorem summarizes our result:

\begin{theorem}
Given a convex polytope $P$ with $n$ vertices and an input parameter $\epsilon \in (0, 1)$, the preprocessing algorithm assigns a unique label of size $O((\lg{(\min(n, \frac{1}{\epsilon}))})^2)$ bits to each vertex of $P$ and it computes a routing table at every vertex of $P$ of amortized size $O((\min(n, \frac{1}{\epsilon^{3/2}}))\lg{n})$ bits in $O(n \min(n^2,$ $\frac{1}{\epsilon^7} \lg{n}))$ time so that any packet is routed along a path on $\partial P$ such that the length of the routing path from any vertex $s$ of $P$ to any vertex $t$ of $P$ is upper bounded by $\frac{8+\epsilon}{\sin{\theta_m}}(D+d(s, t))$, while the size of routing information in the header of any packet is $O((\lg{(\min(n, \frac{1}{\epsilon}))})^2)$ bits.
Here, parameter $D$ is the maximum length between any two points of any cell when $\partial P$ is partitioned into $\frac{1}{\epsilon^3}$ geodesic cells,
$\theta_m$ is half the minimum angle between any two edges of a triangle in the triangulation of $\partial P$, and
$d(s, t)$ is the shortest distance on $\partial P$ from $s$ to $t$.
\end{theorem}

\setcounter{theorem}{0}

Section~\ref{sect:preproc} details the preprocessing algorithm to compute the routing tables at all the vertices of the input convex polytope.
Section~\ref{sect:routing} explains the details of the local routing algorithm and analyzes the same.
Conclusions are in Section~\ref{sect:conclu}.

\section{Preprocessing Algorithm}
\label{sect:preproc}

The input consists of a convex polytope $P$ in $\mathbb{R}^3$ and a real number $\epsilon$ in $(0, 1)$.
The preprocessing algorithm starts by partitioning $\partial P$ into $\delta$-patches.
For any face $f$ of $P$, let $\theta_f^x$ (resp., $\theta_f^z$) be the angle between the normal $n_f$ to $f$ and the vector along the positive $x$-axis (resp., $z$-axis).
For $\delta > 0$, any {\it $\delta$-patch} of $\partial P$ is a maximal collection $\cal C$ of faces of $\partial P$ such that 
(i) $\bigcup_{f \in \cal C} f$ is a contiguous section of $\partial P$, and 
(ii) for any two faces $f_1, f_2 \in \cal C$, both $|\theta_{f_1}^x - \theta_{f_2}^x| \le \delta$ and $|\theta_{f_1}^z-\theta_{f_2}^z| \le \delta$.
Intuitively, each $\delta$-patch is approximately flat, and this approximation is up to the angle $\delta$.
The maximum number of $\delta$-patches of $\partial P$ is dependent on $\delta$, and it does not depend on $n$; the larger the $\delta$, the smaller the number of $\delta$-patches, and vice versa.
From Dudley~\cite{journals/apprxth/Dudley74}, the number of $\delta$-patches of any convex polytope is $O(\frac{1}{\delta^2})$.

Following the above definition of $\delta$-patch, for every pair of integers $i, j \in [1, \frac{2\pi}{\delta}]$, all the faces $f$ of $P$ with $\theta_f^x \in [(i-1) \delta, i \delta]$ and  $\theta_f^z \in [(j-1) \delta, j \delta]$ together form a $\delta$-patch.
Since $P$ is convex, $\theta_f^x$ (resp., $\theta_f^z$) angles of faces $f$ of $P$ are monotonic. 
Hence, any $\delta$-patch so formed is contiguous.
In any iteration, we pick an arbitrary face $f$ of $P$ that is not yet marked to belong to any $\delta$-patch, and find all the faces of $P$ that belong to the same patch as $f$ so that all those faces together form a $\delta$-patch. 
Let $G$ be the dual graph of $\partial P$ wherein each face of $\partial P$ has a unique node in $G$ and $G$ has no other nodes, and two nodes of $G$ are neighbours in $G$ whenever those two faces abut along an edge of $\partial P$.
The following greedy algorithm partitions $\partial P$ into $\delta$-patches: starting at an arbitrary node $v$ of $G$ which does not belong to any $\delta$-patch yet, by applying a breadth-first traversal sourced at $v$ to nodes of $G$ whose corresponding faces of $P$ are yet to be included in any $\delta$-patch, we compute the $\delta$-patch to which the face corresponding to $v$ belongs. 
Since any $\delta$-patch is contiguous, breadth-first traversal would be successful in outputting all faces belonging to a $\delta$-patch.
This algorithm takes $O(n+\frac{1}{\delta^2})$ time to partition $\partial P$ into a set of $O(\frac{1}{\delta^2})$ $\delta$-patches.
For a positive constant $c$, for every $i \in [1, \frac{c}{\delta^2}]$, we denote a $\delta$-patch of $\partial P$ by $\delta_i$.

\begin{lemma}
\label{lem:numpatches}
The $\partial P$ is partitioned into $O(\frac{1}{\delta^2})$ patches, and computing all the $\delta$-patches of $\partial P$ takes $O(n+\frac{1}{\delta^2})$ time.
\end{lemma}

For every $\delta_i$, we pick an arbitrary face $f$ belonging to $\delta_i$.
Let $h_i$ be the half-space defined by the plane $\gamma_i$ containing $f$ such that $h_i$ contains $P$.
We call $f$ a {\it representative face} of $\delta_i$.
Then, $P' = \bigcap_i h_i$ is a convex polytope, and $P'$ contains $P$.
\begin{figure}[ht!]
\centering
\includegraphics[width=5cm]{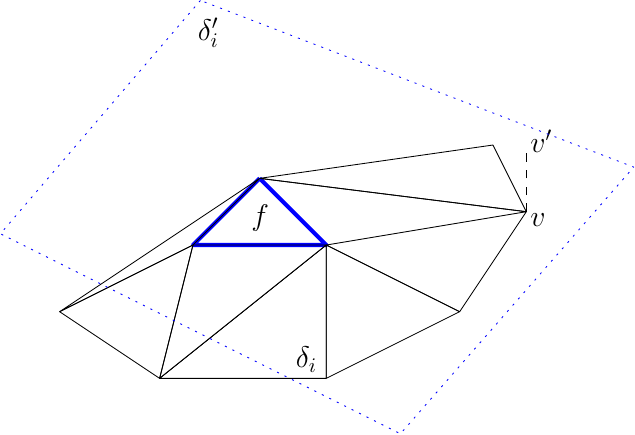}
\caption{\footnotesize 
Illustrates patch $\delta_i$ consisting of a contiguous section of faces from $P$, a representative face $f$ of $\delta_i$, and the plane containing $\delta_i'$.
Also illustrated the point of projection $v'$ on $\delta_i'$ for a vertex $v$ of $\delta_i$.
\normalsize}
\label{fig:patchrep}
\end{figure}
We do not explicitly compute $P'$, but $P'$ and its faces are useful in describing our algorithm.
The face of $P'$ that contains the representative face $f$ of $P$ is denoted by $\delta_i'$.
Let $V_i$ be the set comprising all the vertices of $P$ belonging to $\delta_i$.
For every $\delta_i$, every vertex $v \in V_i$ is projected on the plane containing $\delta_i'$ such that the line segment joining $v$ and the point of projection of $v$ on the plane containing $\delta_i'$ is orthogonal to $\delta_i'$.
Refer to Fig.~\ref{fig:patchrep}.

For every $\delta_i'$, we compute a geodesic spanner $G_i'$ on $\delta_i'$.
That is, every edge of $G_i'$ is located on $\delta_i'$.
Let $V_i'$ be the set comprising all points projected on $\delta_i'$.
To reduce dependency on $n$ in computing $G_i'$, we deterministically sample points in $V_i'$.
Let $\cal B$ be a smallest (rectangular) bounding box containing $\delta_i'$ in the plane containing $\delta_i'$.
We partition $\cal B$ into a set $\cal R$ of squares wherein each of these squares has area equal to the area of $\delta_i'$ multiplied by $\epsilon$.
Noting $\delta_i'$ is convex, for every $r \in \cal R$, if $r \cap \delta_i' \ne \emptyset$, then we call $r \cap  \delta_i'$ is a grid cell on $\delta_i'$.
This results in partitioning $\delta_i'$ into $\frac{1}{\epsilon}$ grid cells.
Then, every point in $V_i'$ belongs to a grid cell of $G_i'$, with ties broken arbitrarily.
The cells of grid $\delta_i'$ are indexed with $j$ where $j \in [1, \frac{1}{\epsilon}]$, and each grid cell is denoted by $g_{ij}'$.
Among all the projected points that belong to any grid cell $g_{ij}'$, an arbitrary point $r_{ij}'$ is chosen as the {\it representative of grid cell} $g_{ij}'$.
The vertex $r_{ij}$ that got projected to $r_{ij}'$ is called a {\it representative vertex} of $\delta_i$.
Refer to Fig.~\ref{fig:selreps}.
We denote the set comprising all the representative vertices on $\delta_i$ by $R_i$.
And the set $R_i'$ denotes the set comprising points resulting from projecting every point $p$ in $R_i$ on $\delta_i'$ such that the line joining the point of projection of $p$ on $\delta_i'$ and $p$ is orthogonal to $\delta_i'$.
Noting that a vertex $v$ of a $\delta$-patch could belong to more than one patch when $v$ is located on the boundary of that $\delta$-patch, $v$ being a representative or non-representative vertex is patch-specific.
However, we call a vertex $v$ of $P$ a representative vertex whenever $v$ is a representative vertex on at least one patch.
Otherwise, $v$ is a non-representative vertex.
The multiset $\bigcup_i R_i$ comprises all the representative vertices of $P$.
Later we set $\delta$ in $\delta$-patches to $\epsilon$.
Then, from Lemma~\ref{lem:numpatches}, the proof of the following is immediate.

\begin{figure}[ht!]
\centering
\includegraphics[width=5cm]{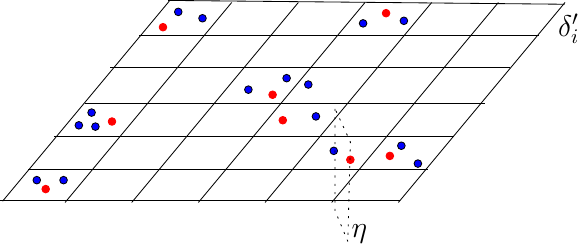}
\caption{\footnotesize 
Illustrates selecting one representative (red) per cell of a grid set up on $\delta_i'$.
Also illustrated is a plane $\eta$ orthogonal to $\delta_i'$ containing the representative (red) and a non-representative (blue) belonging to the same grid cell.
\normalsize}
\label{fig:selreps}
\end{figure}

\begin{lemma}
\label{lem:numrepcomp}
The total number of representative vertices of $P$ is $O(\min(n, \frac{1}{\epsilon^3}))$, and computing the set comprising all the representative vertices of $P$ takes $O(n+\frac{1}{\epsilon^3})$ time.
\end{lemma}

For every $i$, for every two points $v_1', v_2' \in R_i'$ with their corresponding vertices $v_1, v_2$ located on $\delta_i$, in the entry corresponding to $v_2$ (resp., $v_1$) in the routing table of $v_1$ (resp., $v_2$), we store the plane orthogonal to $\delta_i'$ that contains both $v_1$ and $v_2$.
Let $v$ be any non-representative vertex on any patch $\delta_i$. 
Suppose the vertex $v'$ on $\delta_i'$ corresponding to $v$ on $\delta$ belongs to grid cell $g_{ij}'$.
Also, let $r_{ij}'$ be the representative of $g_{ij}'$ which resulted from projecting $r_{ij}$ on $\delta_i$ to $\delta_i'$.
Then, in the entry corresponding to $v$ in the routing table at $r_{ij}$, we store the plane $\eta$ orthogonal to $\delta_i'$ that contains both $v$ and $r_{ij}$.
Refer to Fig.~\ref{fig:selreps}.
Additionally, in the routing table at $v$, we store $\eta$ to help in reaching its representative $r_{ij}$.

Next, following Kapoor and Li~\cite{conf/isaac/KapoorL09}, for every $i$, we compute a \textTheta-graph based geodesic Steiner spanner $G_i'$ whose node set is a superset of $R_i'$, and $G_i'$ provides a spanner path between every two nodes of $G_i'$.
The node set of $G_i'$ is $R_i' \cup S_i'$, where $S_i'$ is the set of Steiner nodes introduced into $G_i'$ as the algorithm proceeds in computing $G_i'$.
Next, we describe the set $S_i'$ and how it is computed.
The Steiner nodes are useful in computing routing entries at the vertices of $P$, and we do not store routing tables at Steiner nodes.
The node set of $G_i'$ is initialized with $R_i'$.
Let $o$ be an arbitrary point located on $\delta_i'$.
Let $\cal C$ be a set of $\frac{2\pi}{\epsilon}$ cones, each with cone angle $\epsilon$, together partitioning $\delta_i'$ such that the apex of each of these cones is at $o$.
Also, let $S_j'$ be the set comprising Steiner nodes in the geodesic spanner computed on $\delta_j'$, wherein $\delta_j'$ abuts $\delta_i'$ along an edge of $P'$.
For every $p' \in R_i'$ and for every cone $C \in \cal C$, let $C_{p'}$ be the cone obtained by translating $C$ so that the apex of $C$ is at $p'$.
We introduce a Steiner point on the $\partial \delta_i' \cap C_{p'}$ that is closest to $p'$ as a Steiner node in $S_i'$ as well as in every $S_j'$.
By similar means, Steiner points could also be introduced on the boundary of $\delta_i'$ due to the application of the algorithm on the faces of $P'$ that are neighbours of $\delta_i'$.
In constructing the geodesic spanner $G_i'$ based on \textTheta-graphs, we include Steiner points from the neighbouring faces of $\delta_i'$ as well into $S_i'$.
These are the only Steiner points that we introduce into $G_i'$.
Since we rely on \textTheta-graph construction on each face of $P'$, the Steiner points introduced in this way suffice to yield an $(1+\epsilon)$-approximation of the geodesic distance between any two representative or Steiner points on the $\partial P'$.

Any line segment $\ell'$ on a face of $P'$ when projected to $\partial P$ is a path, a sequence of line segments on $\partial P$, and the length of this path is upper-bounded due to the patch subdivision of $\partial P$.
Since the endpoints of $\ell'$ result from geodesically projecting two vertices $x, y$ of $P$, projecting $\ell'$ to $\partial P$ yields a path with endpoints $x$ and $y$. 
Following Clarkson~\cite{conf/stoc/Clarkson87}, for every $i$, we introduce edges into $G_i'$:
For every $p' \in R_i' \cup S_i'$, we consider $C_{p'}$.
We project every point $p'' \in R_i' \cup S_i'$ that is located in $C_{p'}$ on the line $\ell_{C_{p'}}$ that bisects the cone angle of $C_{p'}$.
Among all such projected points on $\ell_{C_{p'}}$, let $q''$ be the point that is closest to $p'$.
Also, let $q'$ be the point in $C_{p'}$ whose projection on $\ell_{C_{p'}}$ resulted in $q''$.
Then, we introduce an edge between $p'$ and $q'$ into $G_i'$.
Note that the edges introduced in this way integrate all the Steiner points introduced into the geodesic spanners being computed.
The geodesic Steiner spanner $G'$ on $\partial P'$ is $\bigcup_i G_i'$, and $G'$ helps in building routing tables at the set $\bigcup_i R_i$ of representative vertices of $P$.

Given any input set $S'$ comprising $n'$ points in the plane and a real number $\epsilon \in (0, 1)$, the algorithm given in \cite{conf/stoc/Clarkson87} computes a \textTheta-graph of size $O(\frac{n'}{\epsilon})$ in $O(\frac{n'}{\epsilon} \lg{(\frac{n'}{\epsilon})})$ time so that the stretch of the resulting graph is $1/(\cos{\epsilon}-\sin{\epsilon}) \approx (1+\epsilon)$.
Noting $\lg{(\min(\frac{n}{\epsilon}, \frac{1}{\epsilon^4}))} = \lg{(\frac{1}{\epsilon})}$, the following lemma is immediate from both Lemma~\ref{lem:numrepcomp} and the result from \cite{conf/stoc/Clarkson87}.

\begin{lemma}
\label{lem:G'size}
The number of nodes in the geodesic spanner $G'(R' \cup S', E')$ is $O(\min(n, \frac{1}{\epsilon^3}))$ and the number of edges is $O(\min(\frac{n}{\epsilon}, \frac{1}{\epsilon^4}))$.
Given the representative points on $\partial P'$, computing $G'$ takes $O(\min(\frac{n}{\epsilon}, \frac{1}{\epsilon^4})\lg{(\frac{1}{\epsilon})})$ time.
For any two representative points $u', v'$ on $\partial P'$ with the length of the shortest path between $u'$ and $v'$ on $\partial P'$ being $d_{\partial P'}(u', v')$ and the length of the shortest path in $G'$ between $u'$ and $v'$ being $d_{G'}(u' v')$, $d_{\partial P'}(u', v') \le d_{G'}(u', v') \le (1+\epsilon) \cdot d_{\partial P'}(u', v')$.
\end{lemma}

We further process $G'$ to compute routing table entries needed to route packets between a representative vertex on $\delta_i$ to a representative vertex on $\delta_j$, for $i \ne j$.
By applying the preprocessing algorithm in Thorup and Zwick~\cite{conf/spaa/ThorupZwick01} to $G'$, we compute routing tables at the nodes of $G'$.
We consider every two representative vertices $u, v$ of $P$ respectively belonging to distinct patches $\delta_i$ and $\delta_j$ with their corresponding nodes $u', v' \in G'$. 
We first note that since there is no edge between $u'$ and $v'$ in $G'$, our routing algorithm cannot forward the packet from $u'$ to $v'$.
Hence, in routing a packet from $u'$ to $v'$ using the packet headers as in \cite{conf/spaa/ThorupZwick01}, suppose the packet gets forwarded to node $w'$ of $G'$ and $w$ corresponding to $w'$ is located on $\delta_i$. 
Then, in the routing table at $u$, for the entry corresponding to routing to $v$, we store a plane $\eta$ orthogonal to $\delta_i'$ while $\eta$ contains both $u'$ and $w'$. 
In the other case, wherein $w' = \delta_i' \cap \eta$ is a node in $G'$ that corresponds to a Steiner point on an edge $xy$ of $P$.
The vertices of each such edge are called {\it marked vertices} of $P$.
For every Steiner point $p$ on any edge of $\partial P$ with endpoints $x$ and $y$, we mark both $x$ and $y$ so that these two vertices take the responsibility of forwarding the packets reaching $p$ further.
Both at $x$ and at $y$, we store the plane $\beta$ orthogonal to $\delta_j'$, where edge $xy \in \delta_j \cap \delta_i$ and the packet got routed to $p$ along a path on $\delta_i$ that passes through $p$. 
As explained earlier, the plane $\beta$ saved at $x$ (resp., $y$) helps guide the packet to another Steiner point on $\delta_j$ or to a vertex of $\delta_j$.
This maintains the invariant that successive vertices along any routing path belong to the same patch of $P$.
The sequence of these planes together defines an approximate shortest path on $\partial P$ between any two representative vertices $s$ and $t$ of $P$, with some of the intermediate nodes being points on $\partial P$ corresponding to Steiner nodes of $G'$.
This path is further approximated by a routing path that passes through one of the marked nodes corresponding to each Steiner point of $P$ along it. 
At the end of computing the routing tables at the nodes of $P$, there is no use of $G'$, and hence the data structures augmented to nodes of $G'$, including the ones due to \cite{conf/spaa/ThorupZwick01}, are destroyed. 

Now we describe routing tables at the end of the preprocessing phase.
The vertex set of $P$ is partitioned into a set of representative vertices and a set of non-representative vertices.
A vertex $v$ is a representative vertex of $P$ if the projection $v'$ of $v$ on $\partial P'$ is chosen as the representative of a grid cell located on a face of $P'$.
Otherwise, $v$ is a non-representative vertex.
Each non-representative vertex $v$ of $P$ is associated with a representative vertex $r_v$ of $P$ such that their projections $v'$ and $r_v'$ on $\partial P'$ are on the same face $f'$ of $P'$ and belong to the same grid cell of $f'$.
This does not necessarily ensure $v$ and $r_v$ belong to the same face $f$ of $P$; however, both $v$ and $r_v$ belong to the same patch of $P$.
Hence, for any packet at $v$ whose destination is not $v$, we simply forward the packet to $r_v$.
That is, each non-representative vertex $v$ of $P$ can only forward any packet at $v$ to its representative $r_v$.
To facilitate this, two rays each originating at $v$ are stored at $v$ so that these two rays together define a plane orthogonal to $f'$ that contains $v$.
Analogously, to route from any representative vertex $r_v$ of $P$ to any vertex that belongs to the set $W$ comprising non-representative vertices of $P$ whose representative is $r_v$, we store two rays corresponding to each node $v \in W$ in the routing table at $r_v$.
These two rays help in routing the packet from $r_v$ to $v$.
To route between any two representative vertices $v$ and $w$ that belong to the same patch of $P$, again we store two rays at $v$ (resp., $w$) which together help in routing the packet from $v$ (resp., $w$) to $w$ (resp., $v$).
In either of these cases, the plane defined by these two rays intersects a contiguous sequence $\cal F$ of faces belonging to a patch.
As will be explained in the routing algorithm, the routing path passes through a subset of vertices of faces in ${\cal F}$.
In the last case, a routing path is determined from a representative vertex $v$ to another representative vertex $w$, and these two vertices belong to distinct patches of $P$.
Such a path could intersect multiple patches and pass through Steiner points on the edges of $P$.
To facilitate passing a packet via a Steiner point located on any edge $xy$, as described above, we take the help of routing tables at the marked vertices $x$ and $y$.

In addition, for every two faces $f$ and $g$ sharing an edge $xy$ with the vertices of $f$ being $u, x, y$ and the vertices of $g$ being $v, x, y$, we store both the coordinates and the label of $v$ (resp., $u$) at $u$ (resp., $v$).
Note that $u, x, y$ together define the face $g$ that shares an edge $e$ opposite to $u$.
The routing algorithm details the need to save this information.
Though $v$ is not a neighbour to $u$ on the $\partial P$, we consider $v$ as local to $u$ since it is a $2$-neighbor to $u$, that is, there is a path between $u$ and $v$ along $\partial P$ with one intermediate vertex of $P$.
The routing path computed in the routing phase is such that every two successive vertices along that path are the vertices of an edge of $P$.

\begin{lemma}
\label{lem:amsizetable}
The amortized size of the routing table at any vertex of $P$ is $O((\min(n,$ $\frac{1}{\epsilon^{3/2}})) \lg{n})$ bits.
\end{lemma}
\begin{proof}
As detailed above, since each non-representative vertex $v$ is associated with exactly one representative $r_v$, the number of entries required to store the representative at any non-representative vertex of $P$ is $O(1)$.
Each such entry is of size $O(\lg{n})$, as this involves storing the label of $r_v$, the coordinates of $r_v$, and the coordinates of a plane.
Further, each representative vertex $r_v$ stores only routing entries corresponding to vertices of $P$ that belong to the grid cell on which $r_v$ is located.
Again, any such entry is of size $O(\lg{n})$ bits.
There are $O(\min(n, \frac{1}{\epsilon^3}))$ representative vertices in total.
At every representative vertex $v$ of $P$, every neighbour $w$ of $v$ on $\partial P$ and its projection $w'$ are saved with $v$.
Due to routing algorithm in \cite{conf/spaa/ThorupZwick01}, the routing table computed at every vertex of $P$ is of size $O(\min(n, \frac{1}{\epsilon^{3/2}}) \lg{n})$ bits.
Each Steiner point on the boundary of any two patches that abut causes two marked vertices of $P$, and the coordinates of a plane are stored at each such marked vertex corresponding to a Steiner point.
Considering the sizes of routing tables due to \cite{conf/spaa/ThorupZwick01}, the number of Steiner points caused by any vertex of $P$ is $O(\min(n,$ $\frac{1}{\epsilon^{3/2}}))$.
In addition, at every vertex $v$ of $P$, for every face $f$ incident to $v$ with edge $xy$ shared between $f$ and another face $g$ of $P$ wherein $g$ is defined with vertices $w, x,$ and $y$, using $O(\lg{n})$ amortized space, we store the coordinates and the label of $w$ in the routing table at $v$.
Further, we note that none of the routing tables stores any information related to \textTheta-graphs.
\end{proof}

\begin{lemma}
\label{lem:preprtime}
The time to preprocess $P$ is $O(n \min(n^2, \frac{1}{\epsilon^7} \lg{n}))$.
\end{lemma}
\begin{proof}
We consider the time required for each step of the preprocessing phase.
By taking $\delta$ equal to $\epsilon$, from Lemma~\ref{lem:numpatches}, the time to partition $\partial P$ into $\epsilon$-patches takes $O(n+\frac{1}{\epsilon^2})$ time.
Projecting any vertex of a patch $\delta_i$ on $\delta_i'$ takes $O(1)$ time.
From Lemma~\ref{lem:numrepcomp}, the time for setting up grid cells over the faces of $P'$ and partitioning the vertex set of $P$ into representatives and non-representatives takes $O(n+\frac{1}{\epsilon^3})$ time.
For every $i$, computing a geodesic spanner on the face $\delta_i'$ of $P'$, including the computations of \textTheta-graphs together takes $O(\min(\frac{n}{\epsilon}, \frac{1}{\epsilon^4})\lg{(\frac{1}{\epsilon})})$ time.
The computation involved in finding planes containing a non-representative vertex and its corresponding representative vertex, and vice versa, together takes $O(n)$ time.
The number of marked nodes whose routing tables need to be populated with Steiner points on the boundaries of patches is upper bounded by the \textTheta-graph computations.

Given a graph $T'$ with $n'$ nodes and $m'$ edges, the algorithm in \cite{conf/spaa/ThorupZwick01} computes a routing table of size $O(\sqrt{n'})$ at each node of $T'$ in $O(n'm')$ time.
Hence, due to Lemma~\ref{lem:G'size}, due to the routing algorithm from \cite{conf/spaa/ThorupZwick01}, and since we consider every pair of representative vertices, the time to compute entries of all the routing tables at all the representative and marked vertices of $P$ together takes $O(\min(n^3, \frac{1}{\epsilon^7}))$ time.

At every vertex $v$, for every neighbour $w$ of $v$, saving the identifiers of $w$ and $w'$ takes $O(\lg{n})$ time.
Since any vertex label is of size $O(\lg{n})$ and since each routing table entry has $O(1)$ such labels and coordinates, considering Lemma~\ref{lem:amsizetable}, the stated time complexity includes $O(\lg{n})$ time to write each of those routing table entries.
For every vertex $v$ of $P$ and for every face $f$ incident to $v$ with edge $xy$ shared between $f$ and another face $g$ of $P$, wherein $g$ is defined with vertices $w, x,$ and $y$, finding $w$ and storing both of its coordinates and the label together takes $O(\lg{n})$ time. 
As a whole, the number of entries in all the routing tables and the sizes of those entries together dominate the time complexity.
\end{proof}

\section{Routing Algorithm}
\label{sect:routing}
Given a source vertex $s$ of $P$ and a destination vertex $t$ of $P$, with the help of routing tables computed at the vertices of $P$, the routing algorithm described in this section routes a packet from $s$ to $t$.
As detailed, the routing algorithm needs to handle routing in the following cases: 
routing from a non-representative vertex to its representative, 
routing from a representative vertex $r$ to a vertex whose representative is $r$,
routing between two representative vertices while both of them belong to the same patch, and
routing between two representative vertices while the patches to which they belong are distinct.
Since in the first three cases routing is between two vertices that belong to the same patch, there are essentially two cases: routing from $s$ to $t$ when both $s$ and $t$ are located on the same patch, and routing from $s$ to $t$ when $s$ and $t$ are located on two distinct patches, say $\delta_i$ and $\delta_j$ respectively.

Any vertex $v_1$ forwards the packet to another vertex $v_2$ only when there is an edge of $P$ with endpoints $v_1$ and $v_2$.
Considering this, given two vertices of $P$ and a geodesic path $\pi(s, t)$ from $s$ to $t$ located on $\partial P$, we propose a greedy algorithm that routes from $s$ to $t$ along a {\it routing path} $\pi_r(s, t)$ such that every two successive vertices along the routing path are joined by a triangulation edge of $\partial P$.
Note that any two successive intermediate vertices in $\pi(s, t)$ need not be joined by a triangulation edge of $P$, whereas in $\pi_r(s, t)$ they have to be connected.
We call a vertex $v$ of $P$ a {\it pseudo-destination} of a packet whenever the packet is at some other vertex $u$ of $P$, the routing algorithm decides to route it along the edges of $P$ to $v$ first, as part of routing the packet to its destination $t$.
Once the packet reaches $v$, the routing tables at $v$ determine how to route it to $t$.
Until the packet reaches $v$, vertex $v$ is stored as a pseudo-destination in the packet header.
Two vectors originating at $v$ are stored in the routing table at $v$, which together define a plane $\eta$, and the intersection of $\eta$ with $\partial P$ determines the route of the packet from $v$ onwards. 

Indeed, the concatenation of the ordered sequence of $\pi(v_i, v_{i+1})$ for $i = 0$ to $k-1$, with $v_0 = s$ and $v_k = t$, is $\pi(s, t)$.
Naturally, $|\pi(s, t)| = \sum_{i=0}^{k-1}$ $|\pi(v_i, v_{i+1})|$.
Each of these subpaths is approximated by a routing path; that is, for every $i$, there is a $\pi_r(v_i, v_{i+1})$ corresponding to $\pi(v_i, v_{i+1})$, 
The stretch factor due to our routing algorithm is an upper bound on the worst-case ratio of $|\pi_r(s, t)|$ to $|\pi_{opt}(s, t)|$, where $\pi_{opt}(s, t)$ is an optimal shortest path between $s$ and $t$ on $\partial P$.
Essentially, the objective of the local routing algorithm is to route the packet from one pseudo-destination (or source of the packet) to another pseudo-destination until the packet arrives at its destination.

\pagebreak

\subsection*{\bf Greedy local routing algorithm}

Next, we describe the greedy routing algorithm.
As mentioned, for any two points $p$ and $q$ located on any patch, no packet is forwarded from $p$ to $q$ (or, vice versa) if $p$ and $q$ are not neighbours on $\partial P$.
That is, $p$ can forward a packet to $q$ only if there is an edge on $\partial P$ between $p$ and $q$.
Hence, to route a packet from vertex $p$ to vertex $q$, both located on a patch $\delta_i$, if there is no edge joining $p$ and $q$ on $\partial P$, the routing algorithm needs to find a route on $\delta_i$ such that any two successive vertices along that path are endpoints of an edge of $\delta_i$.
In the packet header, we also save the last vertex visited by the packet.
This is needed for the routing algorithm as detailed below.

\begin{figure}[ht]
\centering
\includegraphics[width=5cm]{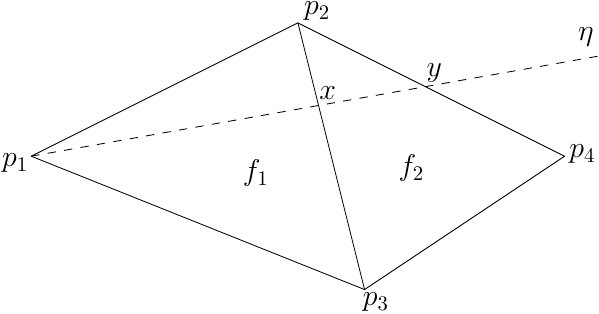}
\caption{\footnotesize 
Illustrates how the next hop for the packet is determined at vertex $p_1$ when the packet is at $p_1$.
In the case shown, since the plane $\eta$ intersects edge $p_2p_4$ of $f_2$, the packet at $p_1$ is forwarded to $p_2$.
\normalsize}
\label{fig:locgreedy1}
\end{figure}

Suppose a packet arrives at a vertex $p_1$ located on patch $\delta_i$ of $P$.
We describe how the packet is routed from $p_1$ to its pseudo-destination.
Based on the pseudo-destination stored in the packet header, the routing algorithm finds the plane $\eta$ from the routing tables at $p_1$.
Let $r_1$ and $r_2$ be the rays originating at $p_1$ which together define $\eta$.
As detailed in the preprocessing phase, $\eta$ is orthogonal to $\delta_i'$, and the pseudo-destination is located on $\delta_i$.
With a binary search over the faces incident to $p_1$, we determine the face $f_1$ intersected by $\eta$ such that the projection of $r_1$ is incident on $f_1$.
Let $p_2$ and $p_3$ be the other two vertices of $P$ defining $f_1$ such that $p_1, p_2, p_3$ occur in that order when the boundary of $f_1$ is traversed in a clockwise direction.
If the point $x = \partial f_1 \cap \eta$ is a vertex of $f_1$, then the routing algorithm forwards the packet to that vertex.
Otherwise, $x$ is a point located on edge $p_2p_3$ of $f_1$.
Refer to Fig.~\ref{fig:locgreedy1}.
Let $f_2$ be a face with vertices $p_2, p_3, p_4$ such that $p_4, p_3, p_2$ occur in that order when the boundary of $f_2$ is traversed in a clockwise direction.
Now there are three cases to consider.
The intersection of $\eta$ with $f_2$ is a point $y$ located on edge $p_2p_4$ or on edge $p_3p_4$, or $y$ is vertex $p_4$ itself.
If $y$ is located on $p_2p_4$ (resp., $p_3p_4$), then the packet at $p_1$ is forwarded to $p_2$ (resp., $p_3$).
If $y$ is $p_4$ then the packet at $p_1$ is forwarded to $p_2$ if $|p_1p_2|+|p_2p_4| < |p_1p_3|+|p_3p_4|$; otherwise, the packet at $p_1$ is forwarded to $p_3$.
Essentially, the routing algorithm decides the next hop from $p_1$ only after knowing the intersection of $\eta$ with a face $f_2$ that is sharing an edge $p_2p_3$ of a face $f_1$ incident to $p_1$, while edge $p_2p_3$ is opposite to $p_1$ in face $f_1$.
Hence, as mentioned in the preprocessing algorithm, at every vertex $v$ of $P$, the preprocessing algorithm also saves the face that shares an edge $e$ opposite to $v$, wherein $e$ is an edge of a face incident to $v$.

\begin{figure}[ht]
\centering
\includegraphics[width=4cm]{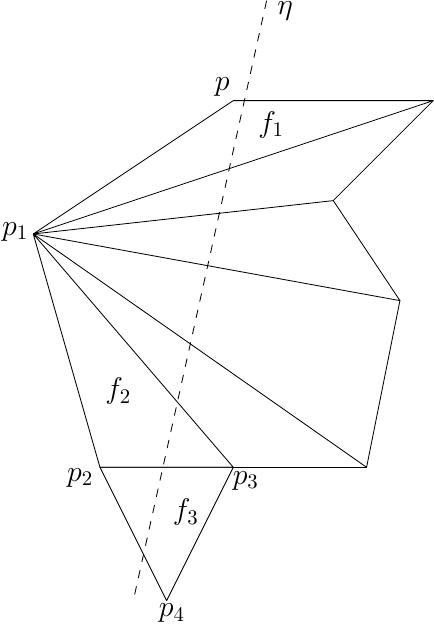}
\caption{\footnotesize 
Illustrates how the next hop for the packet is determined at $p_1$ when the packet is at vertex $p_1$.
Among all the faces incident to $p_1$, in the sequence of faces intersected by $\eta$, faces $f_1$ and $f_2$ are the first and last ones.
Given the packet was at $p$, since $\eta$ intersects $p_2p_4$ of $f_3$ in the case depicted, the packet is forwarded to $p_2$.
\normalsize}
\label{fig:locgreedy2}
\end{figure}

Consider the general case in which the packet needs to be forwarded from a vertex $p_1$ of $P$.
With a binary search over faces incident to $p_1$, we find two faces $f_1, f_2$ incident to $p_1$ such that $f_1$ (resp., $f_2$) has an intersection with $\eta$, and the face neighbouring $f_1$ (resp., $f_2$) that is incident to $p_1$ has no intersection with $\eta$.
Refer to Fig.~\ref{fig:locgreedy2}.
Let $p$ be the previous vertex at which the packet resided.
The identifier of $p$ is saved in the packet header before forwarding the packet to $p_1$.
Given $p$, we determine which of the two candidate faces, $f_1$ or $f_2$, led to $p_1$ having the packet.
That is, whether $p$ is a vertex of $f_1$ or $f_2$.
Without loss of generality, suppose $p$ is a vertex of $f_1$.
Then, the packet is routed to one of the endpoints of the edge of $f_2$ that is opposite to $p_1$.
Let $p_1, p_2, p_3$ be the vertices of $f_2$ such that $p_1, p_2, p_3$ occur in that order when the boundary of $f_2$ is traversed in a counterclockwise direction.
Then, the packet is forwarded from $p_1$ to either $p_2$ or $p_3$.
Let $f_3$ be the face such that edge $p_2p_3$ is on the boundary of both the faces $f_2$ and $f_3$.
Also, let $p_2, p_3, p_4$ be the vertices of $f_3$ such that $p_2, p_3, p_4$ occur in that order when the boundary of $f_3$ is traversed in a clockwise direction.
Analogous to the earlier case, the algorithm determines which one among these two vertices receives the packet, depending on whether $\eta$ intersects $rint(p_2p_4)$, $rint(p_3p_4)$, or the vertex $p_4$.

The header of any packet includes the destination and the pseudo-destination labels.
Let a packet be routed from vertex $p$ located on $\delta_i$ to pseudo-destination $q$ located on $\delta_i$.
Then, two vectors defining the plane $\eta$ orthogonal to $\delta_i'$ containing both $p$ and $q$ are stored in the packet header.
The direction encoded in one of these two vectors guides the packet towards $q$ at each intermediate vertex, eventually reaching a vertex $v'$ bounding the face of $P$ containing $q$.
The line segment $pq$ is the result of projecting an edge of the spanner $G'$ on $\partial P$.
Once the packet reaches $v'$, the routing tables at $v'$ determine both the next pseudo-destination and the corresponding plane, and these fields in the packet header are set accordingly.

\subsection*{\bf Stretch of routing path}

First, we state a few lemmas that help in proving the stretch of the routing path produced by the algorithm.

The proof of the following is due to elementary trigonometry, and this lemma is useful in proving propositions mentioned later.

\begin{lemma}
\label{lem:triangineq}
In any triangle $ABC$, $|AB|+|BC| \leq \frac{|AC|}{\sin{(\frac{\angle{ABC}}{2})}}$.
\end{lemma}
\begin{proof}
Let $R$ be the circumradius of triangle $ABC$.
Then, 
\begin{flalign*}
|AB|+|BC| &\le 2R((\sin{\angle{BCA}})+(\sin{\angle{CAB}})) &&\\
&= 4R \cos{(\frac{\angle{ABC}}{2})} \cos{(\frac{\angle{CAB}-\angle{BCA}}{2})} &&\\
&\le 4R \cos{(\frac{\angle{ABC}}{2})}
\end{flalign*}

\begin{flalign*}
\hspace{0.8in}
&= \frac{(2R)(2\sin{(\frac{\angle{ABC}}{2})} \cos{(\frac{\angle{ABC}}{2})})}{\sin{(\frac{\angle{ABC}}{2})}} &&\\
&= \frac{2R\sin{\angle{ABC}}}{\sin{(\frac{\angle{ABC}}{2})}}  &&\\
&= \frac{|AC|}{\sin{(\frac{\angle{ABC}}{2})}}.
\end{flalign*} 
Note the first inequality is due to $\frac{|AB|}{\sin{\angle{BCA}}} =  \frac{|BC|}{\sin{\angle{CAB}}} = 2R$.
\end{proof}

The following lemma, stated in \cite{conf/isaac/KapoorL09}, relates the Euclidean distance between any two points $p$ and $q$ located on any $\delta$-patch to the geodesic distance between $p$ and $q$.
\begin{lemma}
\cite{conf/isaac/KapoorL09}
\label{lem:euclgeoddist}
For any two points $p$ and $q$ located on any $\delta$-patch $\delta_i$, $d_P(p, q) \ge |pq| \ge d_P(p, q)/(1+2\delta)$, where $d_P(p, q)$ is the geodesic distance between $p$ and $q$ on $\delta_i$ and $|pq|$ is the Euclidean distance between $p$ and $q$ on $\delta_i'$.
\end{lemma}

The following lemma provides an upper bound on the stretch induced by the greedy routing strategy.
\begin{lemma}
\label{lem:routpathstretch}
For any given geodesic path $\lambda$ on $\partial P$, the length of the corresponding routing path is upper bounded by $\frac{|\lambda|}{\sin{\theta_m}}$, where $\theta_m$ is half the minimum angle between two edges of a triangle in the triangulation of any face of $\partial P$.
\end{lemma}
\begin{proof}
Let $\lambda$ be intersecting the sequence $E$ of edges $u_1v_1, v_1u_2, u_2v_2, v_2u_3 \ldots,$ $u_kv_k$ of $P$ in that order, respectively at $v_1'(=u_1), v_1'', v_2', v_2'', \ldots, v_k'(=v_k)$. 
Also, let $\pi_z$ be the path from $u_1$ to $v_k$ following line segments $u_1v_1, v_1u_2, u_2v_2, v_2u_3 \ldots,$ $u_kv_k$.
From Lemma~\ref{lem:triangineq}, for any $i$, $|v_i'v_i| + |v_iv_i''| \le \frac{|v_i'v_i''|}{\sin{(\frac{\angle{u_iv_iu_{i+1}}}{2})}}$ and $|v_i''u_{i+1}| + |u_{i+1}v_{i+1}'| \le \frac{|v_i''v_{i+1}'|}{\sin{(\frac{\angle{v_iu_{i+1}v_{i+1}}}{2})}}$.
Since summing the left-hand sides of these inequalities is the length of $\pi_z$, $|\pi_z| \le \frac{|\lambda|}{\sin{\theta_m}}$.
Considering the way the packet is routed from any vertex in the local routing algorithm, the length of $\pi_z$ upper bounds the length of the routing path.
Hence, the length of any routing path produced by our algorithm is upper bounded by $\frac{|\lambda|}{\sin{\theta_m}}$.
\end{proof}

The following lemma upper bounds the multiplicative stretch caused by the routing algorithm between any two points located on any $\delta$-patch.

\begin{lemma}
\label{lem:stretchlineseg}
For any two points $p$ and $q$ on any $\delta$-patch, the multiplicative stretch of the routing path produced by the greedy routing algorithm is $\frac{(1+2\delta)}{\sin{\theta_m}}$.
Here, $\theta_m$ is half the minimum angle between two edges of a triangle in the triangulation of any face of $\partial P$.
\end{lemma}
\begin{proof}
From Lemma~\ref{lem:euclgeoddist}, for any two points $p, q$ on any $\delta$-patch, the geodesic distance between $p$ and $q$ is upper bounded by $(1+2\delta)|pq|$.
Hence, due to Lemma~\ref{lem:routpathstretch}, the routing path between $p$ and $q$ is upper bounded by $\frac{(1+2\delta)|pq|}{\sin{\theta_m}}$.  
\end{proof}

Let $s$ and $t$ be two not necessarily representative vertices of $P$.
Let $\pi_r(s, t)$ be the path computed by the routing algorithm.
In the worst-case, $\pi_r(s, t)$ is the concatenation of paths $\pi_r(s, u), \pi_r(u, v)$ and $\pi_r(v, t)$, where $u$ and $v$ are representatives such that the representative of $s$ is $u$ and the representative of $t$ is $v$.
Since $s$ and $u$ belong to the same patch, say $\delta_i$, and since the plane orthogonal to $\delta_i'$ containing $s$ and $u$ is saved in the routing table at $s$, the stretch of $\pi_r(s, u)$ is solely due to the stretch induced by the greedy routing algorithm.
The stretch due to the routing algorithm is upper bounded in Lemma~\ref{lem:stretchlineseg}.
Analogously, the stretch of path $\pi_r(v, t)$ is upper-bounded. 
This analysis works even for $\pi_r(u, v)$ if both the representatives $u$ and $v$ belong to the same patch.
Otherwise, the stretch of $\pi_r(u, v)$ relies on the stretch of $G'$, the stretch of the routing path computed by the algorithm in \cite{conf/spaa/ThorupZwick01} to route from $u$ to $v$, stretch due to projecting each edge along that path on $P$, and the stretch caused by our greedy routing algorithm.
The following lemma provides an upper bound on the stretch of $\pi_r(s, t)$ when $s$ and $t$ lie on distinct patches.

\begin{lemma}
For any two vertices $s$ and $t$ of $P$ located on two distinct grid cells, the length of the path computed by the routing algorithm is upper bounded by $\frac{8+\epsilon}{\sin{\theta_m}}(D+|\pi_{opt}(s, t))|)$.
Here, parameter $D$ is the maximum length of the diagonal of any cell when $\partial P$ is partitioned into $\frac{1}{\epsilon^3}$ geodesic cells of equal size, and $\theta_m$ is half the minimum angle between any two neighbouring triangulation edges of $P$ on $\partial P$.
\end{lemma}
\begin{proof}
Since $s$ and $u$ are located on the same patch, due to Lemma~\ref{lem:stretchlineseg}, the multiplicative stretch induced in $\pi_r(s, u)$ is upper bounded by $\frac{(1+2\delta)}{\sin{\theta_m}}$.
That is, for the points of projection $s'$ and $u'$ on $P'$, respectively corresponding to vertices $s$ and $u$ of $P$, this upper bounds the length of the path resulting from projecting the line segment $s'u'$ on the patch of $\partial P$ that contains both $s$ and $u$.
The same upper bound on the multiplicative stretch is applicable to $\pi_r(v, t)$.
Hence, the length of the routing path from $s$ to $u$ and $v$ to $t$ together is upper bounded by $\frac{(1+2\delta)}{\sin{\theta_m}} 2D$.

Since each cone introduced to compute $G'$ has a cone angle of at most $\epsilon$, from the analysis of \textTheta-graphs \cite{conf/stoc/Clarkson87,conf/swat/Keil88} and from Lemma~\ref{lem:G'size}, the multiplicative stretch of $G'$ is $1+\epsilon$.
Hence, a shortest path in $G'$ from $u'$ to $v'$ has a $(1+\epsilon)$ multiplicative stretch.
Thorup and Zwick's routing algorithm~\cite{conf/spaa/ThorupZwick01} applied over $G'$ yields a $3$ multiplicative stretch.
For any line segment on any patch of $P$, due to Lemma~\ref{lem:stretchlineseg}, the path produced by the greedy algorithm causes a stretch of $\frac{1+2\delta}{\sin{\theta_m}}$.
Considering each line segment along the shortest path from $u'$ to $v'$ in $G'$ projected to $P$ lies on a patch, including all the line segments along the routing path from $u$ to $v$ together cause a multiplicative stretch of $\frac{1}{\sin{\theta_m}}(1+2\delta)(1+\epsilon)3$.
(Below analysis accounts for the term $2D$.)
By choosing $\delta = \epsilon$, considering $\epsilon \in (0, 1)$, this is upper bounded by $\frac{1}{\sin{\theta_m}}(3+15\epsilon)$.

Since the geodesic distance between any two points on a grid cell is upper bounded by $D$, the length of the routing path from $s$ to $t$ is upper bounded by 
$\frac{1}{\sin{\theta_m}}(2D + (3+15\epsilon)|\pi_{opt}(u, v)|)$
$\le \frac{1}{\sin{\theta_m}}(2D + (3+15\epsilon)(2D + |\pi_{opt}(s, t))|)$
$= \frac{1}{\sin{\theta_m}}((8+30\epsilon) D + (3+15\epsilon)|\pi_{opt}(s, t))|)$
$\le \frac{8+30\epsilon}{\sin{\theta_m}}(D+|\pi_{opt}(s, t))|)$.
Replacing $\epsilon$ with $\frac{\epsilon}{30}$ yields the bound stated in the theorem statement.
\end{proof}

\begin{theorem}
Given a convex polytope $P$ with $n$ vertices and an input parameter $\epsilon \in (0, 1)$, the preprocessing algorithm assigns a unique label of size $O((\lg{(\min(n, \frac{1}{\epsilon}))})^2)$ bits to each vertex of $P$ and it computes a routing table at every vertex of $P$ of amortized size $O((\min(n, \frac{1}{\epsilon^{3/2}}))\lg{n})$ bits in $O(n \min(n^2,$ $\frac{1}{\epsilon^7} \lg{n}))$ time so that any packet is routed along a path on $\partial P$ such that the length of the routing path from any vertex $s$ of $P$ to any vertex $t$ of $P$ is upper bounded by $\frac{8+\epsilon}{\sin{\theta_m}}(D+d(s, t))$, while the size of routing information in the header of any packet is $O((\lg{(\min(n, \frac{1}{\epsilon}))})^2)$ bits.
Here, parameter $D$ is the maximum length between any two points of any cell when $\partial P$ is partitioned into $\frac{1}{\epsilon^3}$ geodesic cells,
$\theta_m$ is half the minimum angle between any two edges of a triangle in the triangulation of $\partial P$, and
$d(s, t)$ is the shortest distance on $\partial P$ from $s$ to $t$.
\end{theorem}

\section{Conclusions}
\label{sect:conclu}
This paper presented a geodesic local routing scheme for routing between the vertices of a convex polytope $P$.
The input parameter $\epsilon \in (0, 1)$, the preprocessing time, sizes of routing tables at the vertices of $P$, size of labels associated to vertices, and the size of the header of the packet are upper bounded in terms of the number of vertices of $P$ and $\epsilon$.
The stretch of the routing path in the devised scheme depends on $D$, which relies on the local geometry of $P$.
This research approximates $P$ with another convex polytope $P'$ whose complexity is a function of $\epsilon$, and computes a geodesic spanner with vertices of $P$ projected on the faces of $P'$ to help in computing routing tables at the vertices of $P$.
In addition, for computing routing tables at the vertices of $P'$, we used a routing scheme designed for abstract graphs.
These help in carefully choosing the next hop at every vertex along the routing path.
To our knowledge, this is the first result for local routing on polytopes.
Future work could include computing patches of fixed diameter to further reduce the influence of parameter $D$.
Besides, future research could aim to improve the sizes of routing tables and the preprocessing time.
In particular, in parallel to the representatives found for abstract graphs as in \cite{journals/jalgo/Cowen01,conf/spaa/ThorupZwick01}, it may be possible to cluster the vertices of $P$ in each patch using geometric techniques to find representative vertices.
It would be interesting to extend this work to routing between the vertices of a polytope which is not necessarily convex.
Besides, when $\mathbb{R}^3$ consists of obstacles that are convex polytopes, routing between the vertices of those polytopes so that the routing path does not intersect the interior of any of those obstacles could also be an interesting problem to explore.

\subsubsection*{Acknowledgement}

This research of R. Inkulu is supported in part by the National Board for Higher Mathematics (NBHM) grant 2011/33/2023NBHM-R\&D-II/16198.

\bibliographystyle{plain}

\end{document}